\documentclass{article}

\usepackage{tabulary,graphicx,times,caption,fancyhdr,amsfonts,amssymb,amsbsy,latexsym,amsmath,amsthm,lineno}
\usepackage{graphicx}
\usepackage{float}
\usepackage[ruled, linesnumbered,noend,lined,boxed]{algorithm2e}
\usepackage{extraplaceins}
\usepackage{epstopdf}
\usepackage[utf8]{inputenc}
\usepackage{url,multirow,morefloats,floatflt,cancel,tfrupee,textcomp,colortbl,xcolor,pifont}
\usepackage[nointegrals]{wasysym}
\makeatletter

\usepackage{ifxetex}
\ifxetex\else
  \usepackage{dblfloatfix}
\fi

\@ifundefined{subparagraph}{
\def\subparagraph{\@startsection{paragraph}{5}{2\parindent}{0ex plus 0.1ex minus 0.1ex}%
{0ex}{\normalfont\small\itshape}}%
}{}

\def\URL#1#2{\@ifundefined{href}{#2}{\href{#1}{#2}}}

\def\UrlOrds{\do\*\do\-\do\~\do\'\do\"\do\-}%
\g@addto@macro{\UrlBreaks}{\UrlOrds}

\makeatother

\usepackage[paperheight=11in,paperwidth=8.3in,margin=2.5cm,headsep=.7cm,top=2.5cm]{geometry}
\usepackage[T1]{fontenc}

\renewenvironment{abstract}
	{\trivlist\item[]\leftskip0pt\par\vskip4pt\noindent
  	\textbf{\abstractname}\mbox{\null}\\}
	{\par\noindent\endtrivlist}

\def\keywords#1{\par\medskip\par\noindent\textbf{Keywords}: #1\par}

 \date{} \emergencystretch 8pt

\makeatletter
\def\author#1{\gdef\@author{\hskip-\tabcolsep%
	\parbox{\textwidth}{\raggedright\bfseries#1\\[1pc]}}}
\def\address[#1]#2{\g@addto@macro\@author{\\\hskip-\tabcolsep\parbox{\textwidth}{\raggedright%
	\normalsize\normalfont\textsuperscript{#1}#2}}}
\let\addresslink\textsuperscript
\def\correspondence#1{\g@addto@macro\@author{\\\hskip-\tabcolsep\parbox{\textwidth}{\raggedright%
	\vspace*{10pt}\normalsize\normalfont~\\#1~\\[12pt]}}}
\def\email#1{\g@addto@macro\@author{\\\hskip-\tabcolsep\parbox{\textwidth}{\raggedright%
	\normalsize\normalfont Emails: #1}}}

\def\title#1{\gdef\@title{\vspace*{-30pt}%
	\raggedright\textbf{\@journaltitle}~\\%
  \raggedright\bfseries\ifx\@articleType\@empty\vspace*{20pt}\else%
  \vspace*{20pt}\@articleType\vspace*{20pt}\\\fi#1}}
\let\@journaltitle\@empty \def\journaltitle#1{\gdef\@journaltitle{{\normalfont\itshape#1}}}
\let\@articleType\@empty \def\articletype#1{\gdef\@articleType{{\normalfont\itshape#1}}}

\let\@runningHead\@empty \def\RunningHead#1{\gdef\@runningHead{{\normalfont #1}}}

\usepackage{fancyhdr}
\fancypagestyle{headings}{\fancyhf{}
  \fancyhead[R]{\itshape\@runningHead}
  \fancyfoot[C]{\thepage}}
\makeatother

\usepackage[%
	numbers,sort&compress%
  ]{natbib}

\usepackage{float,xcolor}

\begin{document}

\title{Approximation Algorithms for the Freeze Tag Problem inside Polygons}

\author{%
		Fatemeh Rajabi-Alni\addresslink{1}, and
  	Alireza Bagheri\addresslink{1}, and 
    Behrouz Minaei-Bidgoli\addresslink{2}}
\address[1]{Computer Engineering Department, Amirkabir University of Technology (Tehran Polytechnic), Tehran, Iran}
\address[2]{Department of Computer Engineering, Iran University of Science and Technology, Tehran, Iran}


\email{f.rajabialni@aut.ac.ir (F. Rajabi-Alni), ar\_bagheri@aut.ac.ir (A. Bagheri), b\_minaei@iust.ac.ir (B. Minaei-Bidgoli)}%


\maketitle

\begin{abstract}
The freeze tag problem (FTP) aims to awaken a swarm of robots with one or more initial awake robots as soon as possible. Each awake robot must touch a sleeping robot to wake it up. Once a robot is awakened, it can assist in awakening other sleeping robots. We study this problem inside a polygonal domain and present approximation algorithms for it.

\end{abstract}

\section{Introduction}
\label{intro}

\textit{Swarm robotics} (SR) aims at coordinating multiple robots towards cooperatively performing a particular job. SR has various applications for example in target search and tracking \cite{Senanayake,Tang}, medicine \cite{Ahuja}, agricultural mechanization \cite{Albiero}, road freight \cite{Gan}, and autonomous underwater vehicles \cite{Hoeher}. For more discussion, see \cite{Osaba}. Consider the following problem from SR. Given a set of robots denoted by $S=\{s_0,\dots, s_n\}$ with a single initial awake robot $s_0$ and the remaining asleep robots, in the \textit{freeze tag problem} (FTP), we study how to awaken all robots as soon as possible. Note that once an asleep robot has been awakened, it can assist in awakening other asleep robots. Each awake robot should move next to an asleep robot to awaken it.

Observe that $S$ can be modeled as points in some metric space such as vertices of a graph with weighted edges. It is proved that the FTP is NP-hard even on star graphs with an equal number of robots at each leaf, however, a polynomial-time approximation scheme (PTAS) and a greedy approximation algorithm (GA) are given for this case \cite{E.M.Arkin}. It is also proved that the approximation factor of the GA, where each robot arriving at the root awakens the unawakened robot at the shortest edge, has a lower bound of $7/3$, moreover, a $14$-approximation algorithm was proposed for the FTP on star graphs with different numbers of robots at each leaf \cite{E.M.Arkin}.

Some heuristics are studied in \cite{M.O.Sztainberg}. If there exist more than one initial awake robot, we denote the problem by $k$-FTP, where $k$ is the number of initially awake robots. A PTAS is presented for the $2$-FTP in Euclidean space \cite{Moezkarimi}. It is proved that the FTP in the $2$ and $3$ dimensional Euclidean spaces is NP-hard in \cite{Abel} and \cite{Johnson}, respectively. An $O(1)$-approximation algorithm and a PTAS are given for the FTP in Euclidean space in \cite{E.M.Arkin}.

The FTP inside simple polygons was studied in \cite{Lubiw}. Depending on the metric used for the distance between robots inside a polygonal domain $P$, we consider two versions for the FTP: the \textit{geodesic FTP} (GFTP) and \textit{visibility FTP} (VFTP). In the GFTP inside $P$, the distance between every pair of robots $s_i$ and $s_j$ for $0\leq i,j\leq n$ is the \textit{geodesic distance} between $s_i$ and $s_j$, i.e. the length of the shortest path between two robots $s_i$ and $s_j$ inside $P$ denoted by $\pi(s_i,s_j)$ (Figure \ref{fig:1}). For the VFTP, the distance between any two robots $s_i$ and $s_j$ is not symmetric; the distance from $s_i$ to $s_j$, called the \textit{geodesic visibility distance} (GVD) from $s_i$ to $s_j$, is the shortest path that $s_i$ should travel inside $P$ such that $s_j$ can be visible from $s_i$ (Figure \ref{fig:12}). Lubiw and Zeng \cite{Lubiw} proved that the GFTP and VFTP are NP-hard. In this paper, we give approximation algorithms for the GFTP and VFTP.

\begin{figure*}
\vspace{0cm}
\hspace{0cm}
 \includegraphics[width=0.8\textwidth]{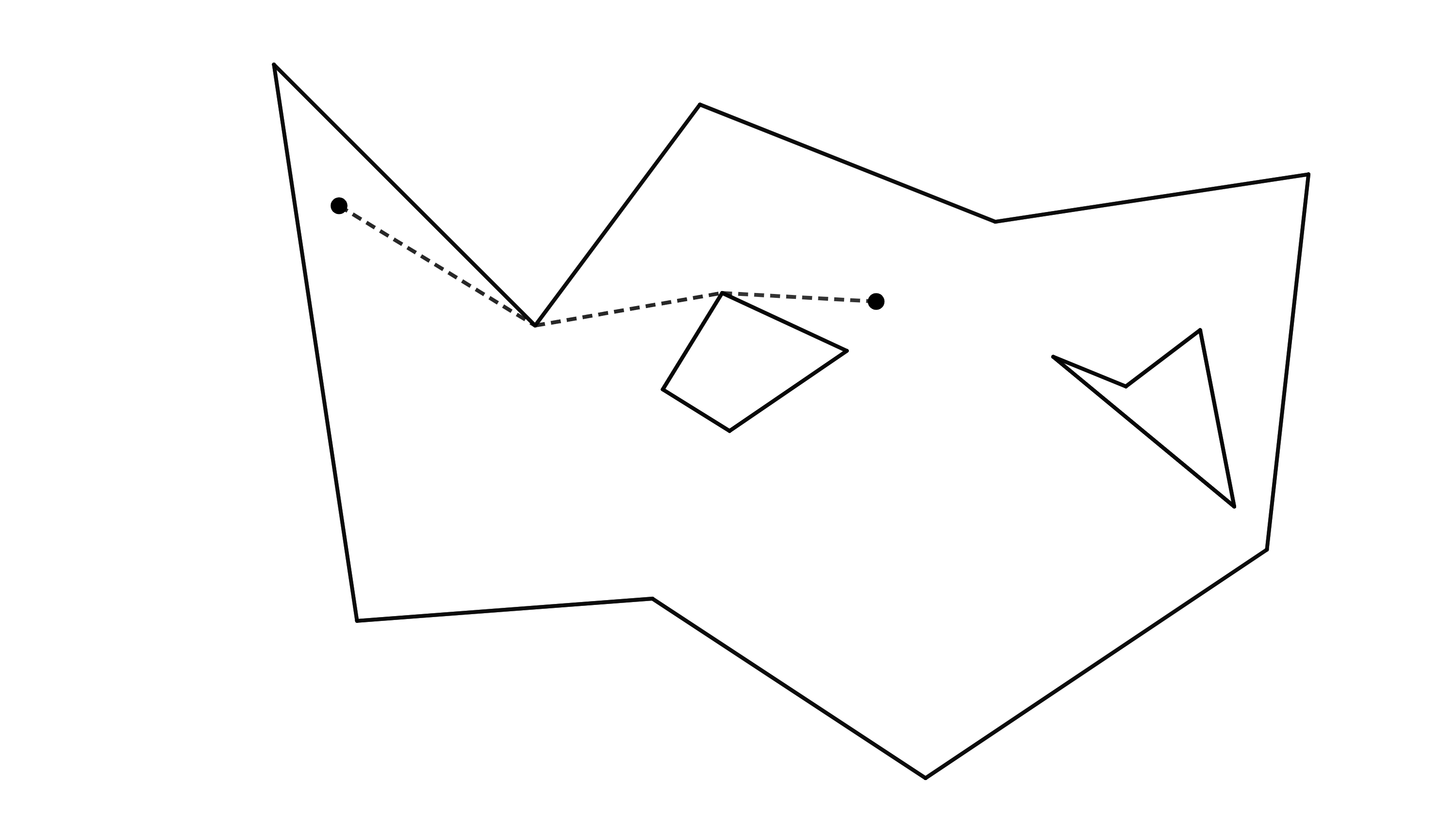}
\vspace{0cm}
\caption{A polygonal domain with two holes; the dashed lines represent the geodesic path between the filled circles.}
\label{fig:1}       
\end{figure*}

\begin{figure*}
\vspace{0cm}
\hspace{0cm}
 \includegraphics[width=0.8\textwidth]{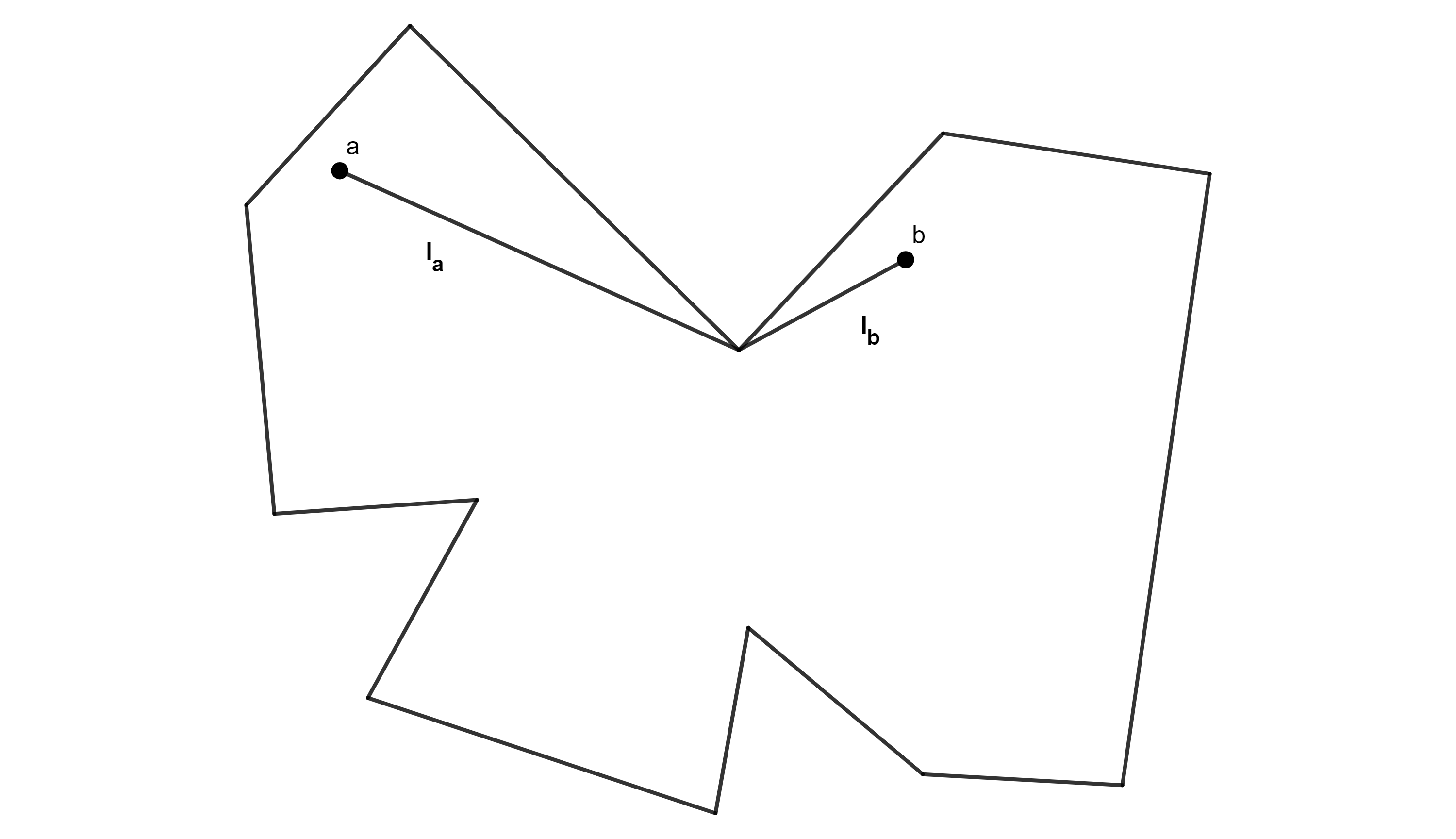}
\vspace{0cm}
\caption{An example for illustration of the GVD; $l_a$ (resp. $l_b$) denotes the GVD from $a$ to $b$ (resp. from $b$ to $a$).}
\label{fig:12}       
\end{figure*}

\section{Our Approximation Algorithms}
In this section, we propose approximation algorithms for two versions of the FTP as follows.
\subsection{Constant Factor Approximation Algorithm}
In this section, we present an $O(1)$-approximation algorithm for the GFTP and VFTP using extra compensatory robots. Let $S$ be a set of $n$ robots (points) inside a polygonal domain $P$. Let $\pi(a,b)$ be the geodesic path between $a$ and $b$ inside $P$ for $a,b \in S$ (Figure \ref{fig:1}). And, let $\|\pi(a,b)\|$ denote the geodesic distance between $a$ and $b$ inside $P$. The \textit{geodesic graph} of $S$ inside $P$ denoted by $GG(S,P)$ is a complete graph with vertex set $S$ such that $e(a,b)=\pi(a,b)$ for $a,b\in S$. Note that $e(a,b)$ denotes the edge between $a$ and $b$.

The \textit{diameter} of $S$ inside $P$, denoted by $diam(S,P)$, is the longest geodesic distance between the robots of $S$ inside $P$, i.e., $diam(S,P)=\max({\|\pi(a,b)\|})_{a,b\in S}$. Let the \textit{makespan} denote the distance traversed from $s_0$ for awakening the last asleep robot. Let $R=\{r_1,r_2,\dots,r_m\}$ denote the set of the reflex vertices of $P$ (if exist). And, let $V=\{v_1,v_2,\dots,v_h\}$ denote the set of the vertices of the polygonal holes inside $P$ (if exist). In each point $v\in R\cup V$, if there does not exist any robot $p\in S$, we locate a new robot called the \textit{Steiner} robot. Let $ST=\{st_1,st_2,\dots,st_{m'}\}$ be the set of the Steiner robots (the robots other than those in $S$) located at the points $R\cup V$ for $m'\leq h+m$.

A \textit{$t$-spanner} of a graph $G=(V',E)$ is a subgraph $H=(V',E')$ of $G$ such that the distance between any two vertices $a,b\in V'$ in $H$ is at most $t$ times the distance between $a$ and $b$ in $G$. Let $S'$ be a set of points in the $d$ dimensional Euclidean space, a \textit{geometric $t$-spanner} of $S'$ is a graph $H'$ with the vertex set $S'$, such that the length of the shortest path between any two points of $S'$ in $H'$ is at most $t$ times the Euclidean distance between them.

\newtheorem{theorem}{Theorem}
\begin{theorem}
\label{theorem1}
Let $S$ be a set of robots inside a polygonal domain $P$ with $m$ reflex vertices and $h$ polygonal hole vertices. There is an $O(1)$-approximation algorithm with the makespan $O(diam(S,P))$ for the GFTP on $S$ using $m'$ Steiner robots, where $m'\leq h+m$.
\end{theorem}

\begin{proof}Let $S'=\{s'_0,s'_1,s'_2,\dots,s'_n\}$ denote a set of robots (points) in the plane. We use the idea of the $O(1)$-approximation algorithm proposed in \cite{E.M.Arkin} for the FTP on $S'$, called the \textit{constant factor approximation algorithm} (CFA) which is described briefly in the following. Note that the distance between any two robots (points) $a,b\in S'$ is the Euclidean distance between $a,b$ denoted by $\|a-b\|$. Firstly, they constructed a degree-bounded $t$-spanner for $S'$ in the plane as follows. They partitioned the plane around each $p\in S'$ into $k$ cones (regions in the plane between two consecutive rays) using rays originating from $p$ at angles $0,{2\pi}/k,2({2\pi}/k),3({2\pi}/k), \dots$.

Assume $u_j(p)\in S'$ denotes the closest robot (point) to the robot (point) $p$ in its $j$th cone. Observe that the robots (points) $u_j(p )$ can be computed in $O(kn\log n)$ total time for all $p\in S'$ and $j$ using standard Voronoi diagrams. Consider the graph $G_K = (S',E_K)$, where there is an edge between each point $p\in S'$ and the nearest neighbors of $p$ in its cones, i.e., the points $u_j(p)$. $G_K$ is called a $\theta$-graph, where $\theta = {2\pi}/k$, and is a $t$-spanner for $k \geq 9$ (and for $k \geq 5$ due to \cite{BOSE2015108}).

Assume that the robots (points) $u_1(p), u_2(p), u_3(p),\dots, u_{k'}(p)$ are in ascending order of their distance from $p$. Thus, $u_{k'}(p)$ is the farthest neighbor of $p$. Once the robot $s'_i\in S'$ is awakened, it starts awakening the robots $u_1(s'_i),u_2(s'_i),\dots,u_{k'}(s'_i)$, respectively, where $u_j(s'_i)$ denotes the $jth$ nearest adjacent vertex of $s'_i$ in $G_k$ for $k'\leq k$. Note that according to the CFA, once $s'_i$ is awakened, it traverses the edge $(s'_i,u_1(s'_i))$ to wake up the robot $u_1(s'_i)$ and travels back to its initial location, then $s'_i$ traverses the edge $(s'_i,u_2(s'_i))$ to wake up the robot $u_2(s'_i)$ and returns to its initial location, and so on.

Observe that if we construct a bounded degree $t$-spanner for a metric, we can give a constant factor approximation algorithm for the FTP in it using the CFA. A geodesic $t$-spanner of $S$ inside $P$ is constructed in \cite{Abam}, but the degrees of the vertices are not bounded by a constant integer $k$. Thus, we first give an algorithm for constructing a degree bounded geodesic $t$-spanner of $S$ inside $P$. We construct a $6$-spanner of the visibility graph of $S\cup R\cup V$ inside $P$ with the degree at most $7$ by the algorithm proposed in \cite{Renssen}.

\newtheorem{lemma}{Lemma}
\begin{lemma}
\label{lem1}
A visibility $t$-spanner of $S\cup R\cup V$ inside $P$ is a $t$-spanner for the geodesic graph of $S$ inside $P$, i.e., $GG(S,P)$.
\end{lemma}

\begin{proof}Let $VT$ be the visibility $t$-spanner of $S\cup R\cup V$ inside $P$ constructed by the algorithm of \cite{Renssen}. Assume that $\pi'(a,b)$ denotes the shortest path between $a,b\in S$ in $VT$.

\newtheorem{observation}{Observation}
\begin{observation}
\label{obser1}
Assume that $\pi(a,b)=a,r'_1,r'_2,\dots,r'_{l},b$ for $a,b\in S$, then we have $r'_j\in R\cup V$ for $1\leq j\leq l$.
\end{observation}


We observe that for every two consecutive vertices of $\pi(a,b)$ denoted by $p$ and $q$, the vertex $p$ is visible from the vertex $q$ (and vice versa). Thus, there exists a path $\pi'(p,q)$ between $p$ and $q$ in $VT$ such that

$$\|\pi'(p,q)\|\leq  t\|\pi(p,q)\|.$$



Therefore, there exists a path $\pi'(a,b)$ between any two robots $a\in S$ and $b\in S$ in $VT$ such that

\begin{equation}
    \nonumber
    \begin{split}
    \|\pi'(a,b)\|\leq &t\|\pi(a,r'_1)\|+t\|\pi(r'_1,r'_2)\|\\
    &+\dots+t\|\pi(r'_{l-1},r'_{l})\|\\
    &+t\|\pi(r'_{l},b)\|.
\end{split}
\end{equation}

And, thus

\begin{equation}
    \nonumber
    \begin{split}
    \|\pi'(a,b)\|\leq &t(\|\pi(a,r'_1)\|+\|\pi(r'_1,r'_2)\|\\
    &+\dots+\|\pi(r'_{l-1},r'_{l})\|\\
    &+\|\pi(r'_{l},b)\|)=t \|\pi(a,b)\|.
\end{split}
\end{equation}\end{proof}


Observe that in the CFA, there exists one robot in every vertex of $G_k$ which starts awakening its adjacent vertices in $G_k$ after it is awakened by another robot. Thus, we use the set of the Steiner robots $ST=\{st_1,st_2,\dots,st_{m'}\}$ with $m'=\|R\cup V\|$ as follows. We apply the CFA on the set of robots $S\cup ST$ with $VT=(S\cup ST,E'')$ as a $t$-spanner. By Lemma \ref{lem1}, there exists a path $\pi'(s_0,s_i)$ in $VT$ such that $$\|\pi'(s_0,s_i)\|\leq t\|\pi(s_0,s_i)\|,$$ for all $s_i\in S$. Consider two consecutive vertices $p,q$ in the path $\pi'(s_0,s_i)$. By the CFA, when the robot $s''_y$ located in $p$ is awakened, it traverses the distance $awakedist(s''_y,s''_z)$ to awaken the robot $s''_z$ located in $q$ for two adjacent robots $s''_y,s''_z\in S \cup ST$ in $VT$.
Let $s''_z=u_j(s''_y)$, i.e., $s''_z$ is the $j$th nearest adjacent vertex of $s'_y$ in $VT$. Note that $\|\pi'(s''_y,u_{j'}(s''_y))\|\leq \|\pi'(s''_y,u_j(s''_y))\|$ for all $j'\leq j$. Thus,
$$awakedist(s''_y,u_j(s''_y))\leq (2j-1)\|\pi'(s''_y,u_j(s''_y))\|.$$
In other words, we have
$$awakedist(s''_y,s''_z)\leq (2j-1)\|\pi'(s''_y,s''_z)\|.$$

Note that we have $j\leq k$, therefore

$$awakedist(s''_y,s''_z)\leq (2k-1)\|\pi'(s''_y,s''_z)\|.$$ Recall that $k$ is the maximum degree of the vertices in $VT$.
Thus, for the distance traversed from $s_0$ for awakening each robot $s_i$, i.e. $awakedist(s_0,s_i)$, we have:
\begin{equation}
    \nonumber
    \begin{split}
    awakedist(s_0,s_i)&\leq (2k-1)\|\pi'(s_0,s_i)\|\\
    &\leq t(2k-1)\|\pi(s_0,s_i)\|.
    \end{split}
    \end{equation}


Note that we have $k=7$ and $t=6$, thus

$$awakedist(s_0,s_i)\leq 78\|\pi(s_0,s_i)\|,$$
for all $s_i\in S$. Therefore, we conclude Theorem \ref{theorem1}.
\end{proof}

As above, we give an $O(1)$-approximation algorithm for the VFTP inside $P$.

\begin{theorem}
\label{theorem2}
For the VFTP on $S\cup ST$ inside $P$, there exists an $O(1)$-approximation algorithm with the makespan $O(diam(S,P))$.
\end{theorem}


\begin{proof}The \textit{geodesic visibility path} (GVP) from $a\in S\cup ST$ to $b\in S \cup ST$ (resp. from $b$ to $a$) is the shortest path that $a$ (resp. $b$) travels inside $P$ until $b$ (resp. $a$) can be visible from $a$ (resp. $b$).

Let $GVP(a,b)$ denote the GVP from $a$ to $b$ in $P$ for $a,b\in S\cup ST$. The \textit{geodesic visibility graph} of $S$ inside $P$, denoted by $GVG(S,P)$, is a directed complete graph with the vertex set $S$ such that $e(a,b)=GVP(a,b)$ and $e(b,a)=GVP(b,a)$ for $a,b\in S$.

\begin{lemma}
\label{lem11}
From the visibility $t$-spanner of $S\cup R\cup V$ inside $P$, we get a $t$-spanner for the geodesic visibility graph of $S$ inside $P$, i.e. $GVG(S,P)$.
\end{lemma}

\begin{proof}Assume $\pi(a,b)=a,r'_1,r'_2,\dots,r'_{l},b$.
Observe that we have:
$$GVP(a,b)=\pi(a,r'_{l}),$$ and
$$GVP(b,a)=\pi(b,r'_1).$$
We observe that $VT$ is a $t$-spanner for $GVG(S,P)$. Since, there exist the paths $\pi'(a,r'_{l})$ and $\pi'(b,r'_1)$ in $VT$ such that
$$\|\pi'(a,r'_{l})\|\leq t\|\pi(a,r'_{l})\|= t\|GVP(a,b)\|,$$ and
$$\|\pi'(b,r'_1)\|\leq t\|\pi(b,r'_1)\|= t\|GVP(b,a)\|.$$
\end{proof}

We use the CFA with $VT$ as a $t$-spanner for $GVG(S,P)$ to awaken the robots in $S\cup ST$. Let $\pi(s''_i,u_j(s''_i))=s''_i,r''_1,r''_2,\dots,r''_{l'},u_j(s''_i)$. Observe that after the robot $s''_i \in S\cup ST$ travels the path $\pi'(s''_i,r''_{l'})$ in $VT$ to awaken the robot $u_j(s''_i)$, it returns to its initial location through the same path, i.e. the path $\pi'(s''_i,r''_{l'})$, although $GVP(s''_i,u_j(s''_i))\neq GVP(u_j(s''_i),s''_i)$ for $s''_i\in S \cup ST$. Thus, as Theorem \ref{theorem1}, we can prove Theorem \ref{theorem2}.
\end{proof}
%

\subsection{A PTAS for the GFTP }
\label{PTAS}

In this section, we give an $O(1+\epsilon)$-approximation algorithm, PTAS, for the GFTP on a set of robots (points) $S$ inside a polygonal domain $P$ using the idea of the PTAS presented in \cite{E.M.Arkin}. Let $S'$ be a set of robots (points) in Euclidean space in any fixed dimension. Assume that one robot is located in each point $p'\in S'$.

In \cite{E.M.Arkin}, a PTAS presented for the FTP in Euclidean space as follows. They first partitioned Euclidean space into $m^2$ subspaces called \textit{pixels}, where $m=1/\epsilon$. Then, they selected an arbitrary robot (point) of each pixel as its representative, and found a \textit{pseudo balanced awakening tree} (SBAT) for the representative robots (points) by examining all possible awakening trees on them in $O(2^{m^2 \log m})$ time. They converted the SBAT to an awakening tree for all robots $p\in S'$, using the awakening tree found by the CFA for each pixel (they replaced the representative of each pixel with the awakening tree of the points inside that pixel). Note that a tree with $n$ vertices is pseudo balanced if the length of each path from the root to a leaf of the tree is $O(\log^2 n)$. The steps of their PTAS are as follows:
\begin{itemize}
  \item Divide the unit square containing the robots into $m^2$ pixels, and select a representative robot for each pixel.
  \item Find an SBAT for representative robots.
  \item Convert the above SBAT to an awakening tree for all robots $p\in S'$ by running the CFA on the points of each pixel.
\end{itemize}

Now, we propose a PTAS for the GFTP on $S$ inside $P$ using the idea of the above PTAS. We assume that there exists one robot in each point $p\in S$. We first decompose the polygonal domain $P$ into convex partitions as follows. Let $SQ(P)$ denote the unit square containing $P$. We divide $SQ(P)$ into $m*m$ squares with side $1/m$ for $m=O(1/ \epsilon)$ (the coordinates of the robots have been rescaled such that they lie in $SQ(P)$). We consider the intersection of $P$ and the squares of $SQ(P)$ as the decomposition (pixels) of $P$ (Figure \ref{fig2}). Note that the number of pixels of $P$ may be less than $m*m$. We observe that the diameter of each pixel of $P$ is $O(1/m)$.
Thus, a PTAS for the GFTP can be given as follows:
\begin{itemize}
  \item We decompose $P$ into $m^2$ or fewer pixels (convex sub polygons), and select an arbitrary robot in each pixel (if exists) as its representative.

  \item We find an SBAT for the representative robots.

  \item Convert the SBAT to an awakening tree of all robots. Observe that each pixel is convex, thus we can use the CFA on the points of each pixel to find an awakening tree for the points in that pixel.
\end{itemize}

\begin{theorem}
\label{theorem3}
There is a PTAS for the GFTP in a polygonal domain $P$.
\end{theorem}

Observe that in the above PTAS, for any two adjacent robots $s''_y,s''_z \in S\cup ST$ in $VT$, either $s''_y$ awakens $s''_z$, or $s''_z$ awakens $s''_y$. Thus, we can assume w.l.o.g. that the geodesic visibility graph of $S$ inside $P$, $GVG(S,P)$, is an undirected graph. Therefore, as above, we can prove the following theorem.

\begin{theorem}
\label{theorem4}
There is a PTAS for the VFTP in a polygonal domain $P$.
\end{theorem}

\begin{figure*}
\vspace{0.5cm}
\hspace{1cm}
 \includegraphics[width=0.9\textwidth]{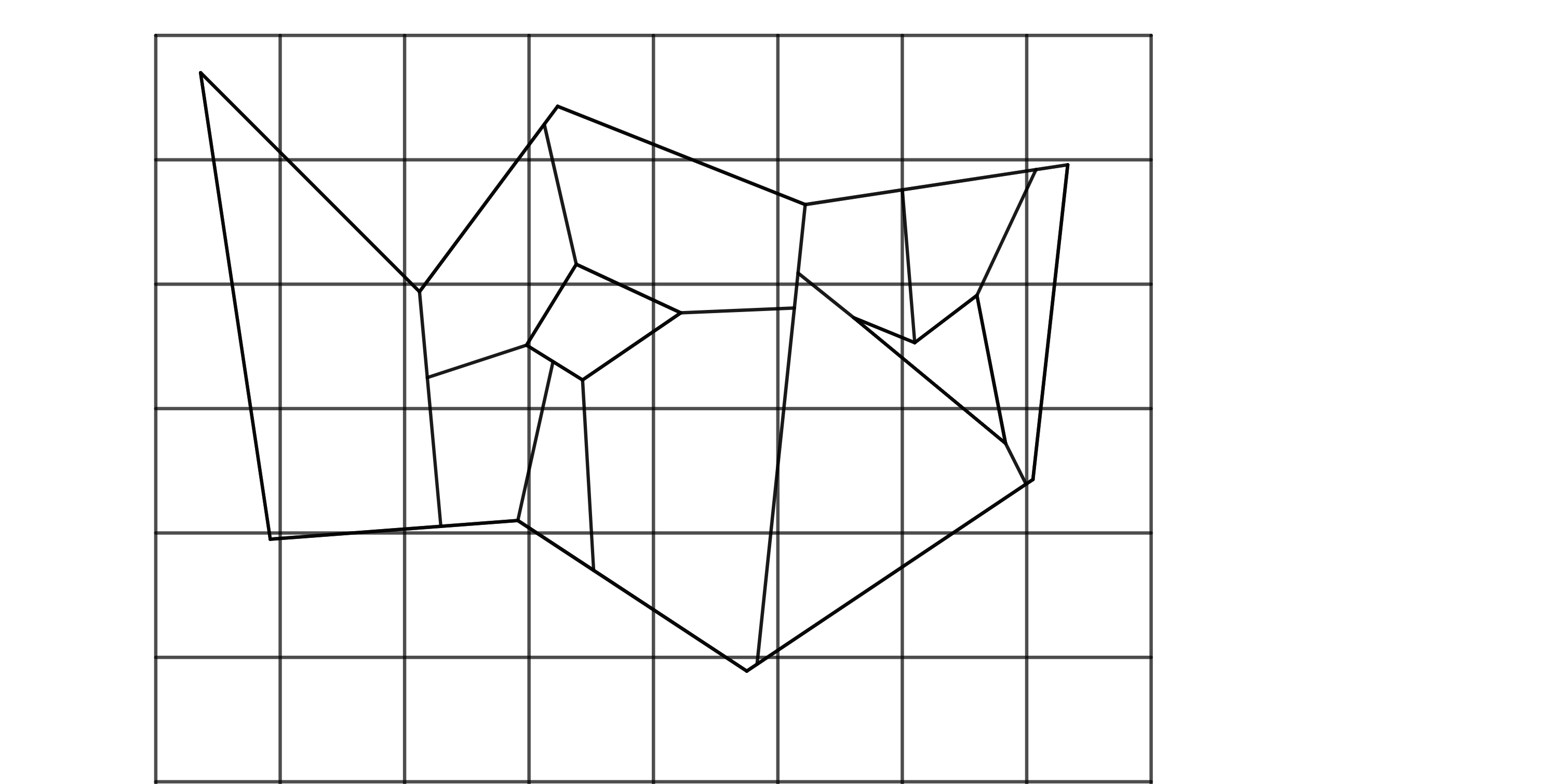}
\vspace{0.4cm}
\caption{An example for decomposing the polygonal domain $P$ into convex partitions (pixels).}
\label{fig2}       
\end{figure*}

\section*{Declarations}

\begin{itemize}
\item[] \textbf{Conflict of interest} The authors have no relevant financial or non-financial interests to disclose.
\item[] \textbf{Data availability} No data were used to support this study.
\item[] \textbf{Authors' contributions} F.Rajabi-Alni designed and proposed the algorithms with support from
A.Bagheri and B.Minaei-Bidgoli. F.Rajabi-Alni and A.Bagheri and B.Minaei-Bidgoli developed the proof and analyzed the
algorithms. F.Rajabi-Alni prepared all ﬁgures. The ﬁrst draft of the manuscript was written by F.Rajabi-Alni. All authors commented on the manuscript. All authors read and approved the ﬁnal manuscript.
\end{itemize}

\FloatBarrier
\bibliographystyle{hindawi_bib_style}
\bibliography{sample6}

\begin{thebibliography}{10}

\bibitem{Senanayake}
M.~Senanayake, I.~Senthooran, J.~C. Barca  et~al.
\newblock ``Search and tracking algorithms for swarms of robots: A survey''.
\newblock {\em Rob. Auton. Syst.}, vol. 75, 422--434, 2016.

\bibitem{Tang}
Q.~Tang, Z.~Xu, F.~Yu, Z.~Zhang  and J.~Zhang.
\newblock ``Dynamic target searching and tracking with swarm robots based on
  stigmergy mechanism''.
\newblock {\em Rob. Auton. Syst.}, vol. 120, 103251, 2019.

\bibitem{Ahuja}
N.~Ahuja, H.~Bikkavilli, Z.~Chen  et~al.
\newblock ``Real-time cellular-level imaging and medical treatment with a swarm
  of wireless multifunctional robots''.
\newblock {\em J. Supercomput.}, vol. 78, 1923–1943, 2022.

\bibitem{Albiero}
D.~Albiero, A.~P. Garcia, C.~K. Umezu  and R.~L. de~Paulo.
\newblock ``Swarm robots in mechanized agricultural operations: A review about
  challenges for research''.
\newblock {\em Comput. Electron. Agric.}, vol. 193, 106608, 2022.

\bibitem{Gan}
M.~Gan, Q.~Qian, D.~Li, Y.~Ai  and X.~Liu.
\newblock ``Capturing the swarm intelligence in truckers: The foundation
  analysis for future swarm robotics in road freight''.
\newblock {\em Swarm Evol. Comput.}, vol. 62, 100845, 2021.

\bibitem{Hoeher}
Peter~Adam Hoeher, Jan Sticklus  and Andrej Harlakin.
\newblock ``Underwater optical wireless communications in swarm robotics: A
  tutorial''.
\newblock {\em IEEE. Commun. Surv. Tutorials}, vol. 23, no. 4, 2630--2659,
  2021.

\bibitem{Osaba}
E.~Osaba, J.~{Del Ser}, A.~Iglesias  and X-S. Yang.
\newblock ``Soft computing for swarm robotics: {N}ew trends and applications''.
\newblock {\em J. Comput. Sci.}, vol. 39, 101049, 2020.

\bibitem{E.M.Arkin}
E.~Arkin, M.~Bender, S.~Fekete, J.~Mitchell  and M.~Skutella.
\newblock ``The freeze-tag problem: How to wake up a swarm of robots''.
\newblock {\em Algorithmica}, vol. 46, no. 2.

\bibitem{M.O.Sztainberg}
Ma.~O. Sztainberg, E.~Arkin, M.~Bender  and J.~Mitchell.
\newblock ``Analysis of heuristics for the freeze-tag problem''.
\newblock In Martti Penttonen  and Erik~Meineche Schmidt, editors, {\em
  Algorithm Theory SWAT}, pages 270--279, Berlin, Heidelberg, 2002. Springer
  Berlin Heidelberg.

\bibitem{Moezkarimi}
Z.~Moezkarimi  and A.~Bagheri.
\newblock ``A {PTAS} for geometric 2-{FTP}''.
\newblock {\em Inf. Process. Lett.}, vol. 114, no. 12, 670--675.

\bibitem{Abel}
Z.~Abel, H.~A. Akitaya  and J.~Yu.
\newblock ``Freeze tag awakening in 2{D} is {NP}-hard''.
\newblock In {\em 27th Annual Fall Workshop on Computational Geometry (FWCG)},
  Stony Brook University, Stony Brook, 2017.

\bibitem{Johnson}
M.~P. Johnson.
\newblock ``Easier hardness for 3{D} freeze-tag''.
\newblock In {\em 27th Annual Fall Workshop on Computational Geometry (FWCG)},
  Stony Brook University, Stony Brook, 2017.

\bibitem{Lubiw}
Y.~Zeng A.~Lubiw.
\newblock ``The visibility freeze-tag problem''.
\newblock In {\em 24th Annual Fall Workshop on Computational Geometry (FWCG)},
  University of Connecticut, 2014.

\bibitem{BOSE2015108}
Prosenjit Bose, Pat Morin, André {van Renssen}  and Sander Verdonschot.
\newblock ``The $\theta_5$-graph is a spanner''.
\newblock {\em Comput. Geom. Theory Appl.}, vol. 48, no. 2, 108--119, 2015.

\bibitem{Abam}
M.~A. Abam.
\newblock ``Spanners for geodesic graphs and visibility graphs''.
\newblock {\em Algorithmica}, vol. 80, 515--529, 2018.

\bibitem{Renssen}
André {van Renssen}  and Gladys Wong.
\newblock ``Bounded-degree spanners in the presence of polygonal obstacle''.
\newblock {\em Theor. Comput. Sci.}, vol. 854, 159--173, 2021.

\end{thebibliography}
\end{document}